\documentclass[a4paper]{article}
\usepackage{a4wide}
\usepackage{authblk}
\usepackage{setspace}
\usepackage{lineno}
\usepackage{endfloat}

\usepackage{crop}%
\usepackage{amsmath,amssymb,amsfonts,upref,endnotes}
\usepackage{amsthm}
\usepackage{marginnote}
\usepackage{soul}
\RequirePackage{graphicx}
\usepackage[figuresright]{rotating}
\usepackage{color}
\makeatletter
\newcommand{\setlabel}[2]{\def\@currentlabel{#2}\label{#1}}
\makeatother
\newtheorem{theorem}{\bf Theorem}[section]

\newcommand{\pd}[2]{\frac{\partial #1}{\partial #2}}

\newcommand{\bv}[1]{\boldsymbol{#1}}
\newcommand{\R}[1]{\mathbb{R}^{#1}}

\begin{document}

%%%% Article title to be placed here
\thispagestyle{empty}

\newpage

\setcounter{page}{1}
\title{
  Biot--Savart law in the geometrical theory of dislocations
}

\author[1]{Shunsuke Kobayashi}
\author[1]{Ryuichi Tarumi}

\affil[1]{Graduate School of Engineering Science, Osaka University, 1-3 Machikaneyama-cho, Toyonaka, Osaka, 560-8531, Japan}
\date{}

\maketitle

%%%%%%%%% Insert author address here
% \address{
%   ${}^{1}$Graduate School of Engineering Science, Osaka University, 1-3 Machikaneyama-cho, Toyonaka, Osaka, 560-8531, Japan
% }

% %%%% Subject entries to be placed here %%%%
% \subject{xxxxx, xxxxx, xxxx}

% %%%% Keyword entries to be placed here %%%%
% \keywords{xxxx, xxxx, xxxx}

%%%% Insert corresponding author and its email address}
% \corres{Ryuichi Tarumi\\
% \email{tarumi.ryuichi.es@osaka-u.ac.jp}}

%%%% Abstract text to be placed here %%%%%%%%%%%%
\begin{abstract}
Universal mechanical principles may exist behind seemingly unrelated physical phenomena, providing novel insights into these phenomena.
This study sheds light on the geometrical theory of dislocations through an analogy with electromagnetics.
In this theory, solving Cartan's first structure equation is essential for connecting the dislocation density to the plastic deformation field of the dislocations.
The additional constraint of a divergence-free condition, derived from the Helmholtz decomposition, forms the governing equations that mirror Amp\`ere's and Gauss' law in electromagnetics.
This allows for the analytical integration of the equations using the Biot--Savart law.
The plastic deformation fields of screw and edge dislocations obtained through this process form both a vortex and an orthogonal coordinate system on the cross-section perpendicular to the dislocation line.
This orthogonality is rooted in the conformal property of the corresponding complex function that satisfies the Cauchy--Riemann equations, leading to the complex potential of plastic deformation.
We validate the results through a comparison with the classical dislocation theory.
The incompatibility tensor is crucial in the generation of the mechanical field.
These findings reveal a profound unification of dislocation theories, electromagnetics, and complex functions through their underlying mathematical parallels.
\end{abstract}

\maketitle
\section{Introduction}

Mathematical universality often underlies the diverse physical phenomena governed by distinct mechanical principles.
The characteristics of the governing differential equations, along with the associated symmetry and conservation laws, intricately connected to the universality of the mechanical system \cite{olver_applications_1986}.
This connection results in remarkable similarities among seemingly unrelated physical phenomena, such as dislocation and vortex \cite{marcinkowski_correspondence_1989,marcinkowski_similarities_1992,lazar_correspondence_2004,silbermann_analogies_2017}, electrostatics and porus media \cite{bar-sinai_geometric_2020}, and elasticity and fractons in quantum systems \cite{pretko_fracton-elasticity_2018,pretko_symmetry-enriched_2018,pretko_crystal--fracton_2019,gromov_colloquium_2024, sun_fractional_2021}.
This approach proves to be invaluable in emerging academic fields, particularly when the governing equations are newly established, and the nature of their solutions remains unclear.
The mechanics of lattice defects in crystalline materials have been subject to extensive theoretical analysis owing to their close connection with structural materials \cite{hull_introduction_2011, anderson_theory_2017}.
This can be traced back to Volterra's classification of dislocations and disclinations \cite{volterra_sur_1907}.
From the 1950s, Kondo \cite{kondo_non-riemannian_1955, kondo_geometry_1955}, Bilby \textit{et al.} \cite{bilby_continuous_1955}, Kr\"oner and Seeger \cite{kroner_nicht-lineare_1959,kroner_allgemeine_1959}, and Amari \cite{amari_primary_1962} proposed a theoretical formulation in non-Euclidean geometry.
More recently, a comprehensive geometrical formulation has been developed using the Riemann--Cartan manifold by Yavari and Goriely  \cite{yavari_riemann--cartan_2012,yavari_riemann--cartan_2013,yavari_geometry_2014}.
In this framework, the intermediate configuration in which only plastic deformation occurs without elastic relaxation, is the most crucial mathematical structure.
This virtual state is determined by solving Cartan's first structure equation (Cartan's equation) for a given dislocation density distribution.
Therefore, analysis of Cartan's equation is essential in gaining profound insights into lattice defects from a geometrical perspective.
In a previous study, we successfully demonstrated that dislocation stress fields result from geometrical frustration within the dislocation core \cite{kobayashi_geometrical_2024}.
This result provides strong evidence for the long-standing mathematical hypothesis --- the duality between curvature and stress \cite{schaefer_spannungsfunktionen_1953,minagawa_riemannian_1962,amari_dual_1963}.
The most critical issue is obtaining an analytical solution to the Cartan's equation.
A mathematically rigorous solution reveals the inherent universality within the mechanics of lattice defects.

From this perspective, various previous studies have delved into the analysis of Cartan's first structure equation, or mathematically equivalent equations, to determine plastic deformation fields.
In the context of geometric dislocation theory, Yavari and Goriely \cite{yavari_riemann--cartan_2012} and Clayton \cite{clayton_defects_2015} have utilised the semi-inverse method, and Edelen \cite{edelen_gauge_1988} and Acharya \cite{acharya_model_2001} employed the homotopy operator to obtain the analytical expression of the plastic deformation field for axially symmetric distributions of dislocations.
In dislocation field theories, analyses on partial differential equations that are essentially equivalent to Cartan's first structure equation have been performed.
Acharya \textit{et al.} \cite{acharya_driving_2003,roy_finite_2005} derived the equations by applying Helmholtz decomposition and incorporating the divergence-free condition for the plastic deformation field, resulting in analytical solutions for a single dislocation \cite{ acharya_structure_2019}.
Lazar \textit{et al.} developed a field theory in which the strain gradient was considered, leading to analytical solutions for the plastic deformation field of a single dislocation \cite{lazar_elastoplastic_2002,lazar_nonsingular_2003}.
Moreover, a general solution to the elastic deformation field was achieved by de Wit using Green's function method \cite{dewit_theory_1973,dewit_theory_1973-1}.
Additionally, numerical solutions for an arbitrary spatial distribution of dislocations have been obtained by Roy and Acharya \cite{roy_finite_2005}, Arora \textit{et al.} \cite{arora_finite_2020} and in our previous study \cite{kobayashi_geometrical_2024}.
However, understanding the mathematical structure of the plastic deformation field can be a challenging task.
Previous studies have primarily focused on investigating the characteristics of strain and stress fields of dislocations by using the plastic deformation field, neglecting the geometric aspects of the plastic deformation field itself.
In summary, the previously utilised analysis methods for Cartan's first structure equation and equivalent governing equation have failed to provide universal analytical solutions essential for a comprehensive understanding of the plastic deformation field.
Therefore, a new perspective on the mathematical structure of dislocations and their plastic deformation fields can be obtained by constructing new analytical solutions.

This study presents a novel perspective on the theory of lattice defects by integrating principles from electromagnetics.
Although this approach may seem unconventional, it is rooted in the mathematical universality shared by theories of lattice defects in the Riemann--Cartan manifold and steady-state electromagnetics.
Moreover, we discovered that the equation aligns with the Cauchy--Riemann equations in complex function theory.
In the following section, we provide a concise overview of the geometrical theory of lattice defects, highlighting the crucial roles of Cartan's equation and Helmholtz decomposition.
In section 3, we analytically integrate Cartan's equation for plastic deformation caused by dislocations.
First, we demonstrate that Cartan's equation, when combined with Helmholtz decomposition, is mathematically equivalent to Amp\`ere's law and Gauss' law in electromagnetism.
Furthermore, we demonstrate the application of analytical solutions to these electromagnetic equations, particularly the Biot--Savart law, in solving the plastic deformation problem.
In section 4, we further explore this analysis by examining its relationship with the complex function theory.
We establish that Cartan's equation is equivalent to the Cauchy--Riemann equation, indicating that the mathematical essence of plastic deformation can be interpreted as a conformal map.
Moreover, we showcase the construction of a complex potential for plastic deformation caused by dislocations.
Finally, we reveal that the linearised stress fields, derived from the analytical plastic deformation fields, align accurately with those anticipated by the classical Volterra dislocation theory.
This result quantitatively validates the mathematical analysis presented in this study.
Finally, section 5 concludes the paper.

\section{Theory of dislocations on Riemann--Cartan manifold}
\label{sec:theory of dislocations on}
\subsection{Construction of the intermediate configuration}
The geometrical theory of dislocations commences by introducing three distinct configurations in kinematics: the reference $\mathcal{R}=(\mathcal{M}, \bv{g}_\mathcal{R}, \nabla_\mathcal{R})$, intermediate $\mathcal{B}=(\mathcal{M}, \bv{g}_\mathcal{B}, \nabla_\mathcal{B})$, and current $\mathcal{C}=(\mathcal{M}, \bv{g}_\mathcal{C}, \nabla_\mathcal{C})$ configurations \cite{yavari_riemann--cartan_2012, kobayashi_geometrical_2024}.
Each configuration is represented as a Riemann--Cartan manifold, that is, a material manifold $\mathcal{M}$ equipped with a Riemannian metric and an affine connection.
A noteworthy mathematical feature is the utilisation of diffeomorphisms of the material manifold $\mathcal{M}$, which distinguishes these configurations based on the metric and connection they possess.
The intermediate configuration is particularly crucial as it represents a virtual state that only contains information regarding plastic deformation and does not exist within the conventional Euclidean space $\mathbb{R}^3$.
In contrast, the reference and current configurations exist in $\mathbb{R}^3$ owing to the embeddings $\bv{x}$ and $\bv{y}$, respectively.
Consequently, the affine connections of these configurations, $\nabla_{\mathcal{R}}$ and $\nabla_{\mathcal{C}}$, are identified with the Euclidean connection.
Following the standard theory of elasticity, these embedding maps are related through the displacement $\bv{u}$ such that $\bv{y}=\bv{x}+\bv{u}$.
The presentation of this theory follows a standard setup outlined herein.
Unless otherwise specified, we will use the Cartesian coordinate system.
The corresponding orthonormal bases and their dual on the reference configuration $\mathcal{R}$ are denoted as $\partial /\partial \bv{x}^i$ and $d\bv{x}^i$ ($i=1,2,3$).

The geometrical construction of the intermediate configuration $\mathcal{B}$ requires the determination of the metric $\bv{g}_\mathcal{B}$ and connection $\nabla_\mathcal{B}$ for a given distribution of dislocations.
The key concept in this theoretical framework is the identification between the dislocation density tensor $\bv{\alpha}$ and the torsion 2-form $\bv{\tau}^i$ which is inherent in the connection $\nabla_\mathcal{B}$ \cite{kondo_non-riemannian_1955,kondo_geometry_1955, bilby_continuous_1955,kroner_allgemeine_1959,kroner_nicht-lineare_1959,amari_primary_1962,yavari_riemann--cartan_2012}.
Consider a scenario in which dislocations are present in a continuous medium, with their distribution defined by the dislocation density tensor $\bv{\alpha}$ in the reference configuration.
For the Burgers vector $\bv{b}=b^i\partial/\partial \bv{x}^i$ and tangent of the dislocation line $\bv{n}=n_j d\bv{x}^j$, we have \cite{nye_geometrical_1953,yavari_riemann--cartan_2012} 
\begin{align}\label{eq:DislocationDensity-Torsion}
    \bv{\alpha}=\bv{b}\otimes\bv{n},
    \quad
    \bv{\tau}^i=*\bv{\alpha}^i,
\end{align}
where $*$ represents the Hodge star operator with respect to the reference metric $\bv{g}_\mathcal{R}$ and $\bv{\alpha}^i= b^i\bv{n}$ represents the $i$-th component of the dislocation density tensor.
In the geometric dislocation theory, the Weitzenb\"ock connection serves as the connection of the intermediate configuration $\nabla_\mathcal{B}$ \cite{yavari_riemann--cartan_2012, yavari_geometry_2014}.
The geometric compatibility of the Weitzenbo\"ck connection elucidates the relationship between torsion and the orthonormal dual bases $\bv{\vartheta}^i$ of the intermediate configuration as follows:
\begin{align}
  \label{eq:cartan's first vartheta}
    d\bv{\vartheta}^i=\bv{\tau}^i,
\end{align}
where $d$ denotes the exterior derivative.
The aforementioned equation represents Cartan's first structure equation \cite{tu_differential_2017}.
Through the diffeomorphism of the material manifold $\mathcal{M}$, the dual bases are connected by a linear transformation each other.
Consequently, the dual bases $\bv{\vartheta}^i$ can be represented as $\bv{\vartheta}^i=(F_p)^i_j d\bv{x}^j$, where $\bv{F}_p$ denotes the plastic deformation gradient.
By substituting this into equation (\ref{eq:cartan's first vartheta}), the matrix $\bv{F}_p$ can be determined by integration.
The challenge lies in the lack of uniqueness of the solution of equation (\ref{eq:cartan's first vartheta}), which can be resolved by applying the Helmholtz decomposition to the dual bases $\bv{\vartheta}^i$.
Based on previous studies \cite{wenzelburger_kinematic_1998,kobayashi_geometrical_2024}, the dual bases $\bv{\vartheta}=(\bv{\vartheta}^1, \bv{\vartheta}^2, \bv{\vartheta}^3)^T$ of the intermediate configuration are defined from the bundle isomorphism between the tangent bundle $T\mathcal{M}$ and product bundle $\mathcal{M}\times \R{3}$.
Consequently, the bases $\bv{\vartheta}$ can be interpreted as an $\mathbb{R}^3$-valued 1-form that undergoes Helmholtz decomposition, such that $\bv{\vartheta}^i=d\bv{\psi}^i+\bv{\Theta}^i$ \cite{wenzelburger_kinematic_1998,kobayashi_geometrical_2024}, where $\bv{\psi}^i$ and $\bv{\Theta}^i$ are referred to as the exact and dual exact forms, respectively.
The cause of the non-uniqueness of equation (\ref{eq:cartan's first vartheta}) is the exact form $d\bv{\psi}^i$ as its exterior derivative vanishes identically: $dd\bv{\psi}^i\equiv 0$.
However, this form does not contribute to torsion $\bv{\tau}^i$.
Therefore, without loss of generality, the exact form can be set as an identity, such that $\bv{\psi}^i=\bv{x}^i$ \cite{kobayashi_geometrical_2024}.
Consequently, the plastic deformation gradient is solely expressed using only the dual exact form $\bv{\Theta}$ as follows: 
\begin{align}
    \label{eq:Helmholtz}
    \bv{F}_p=\bv{I} + \bv{\Theta}
    = \begin{pmatrix}
      1&0&0\\
      0&1&0\\
      0&0&1
    \end{pmatrix}
    +
    \begin{pmatrix}
      \Theta^1_1&\Theta^1_2&\Theta^1_3\\
      \Theta^2_1&\Theta^2_2&\Theta^2_3\\
      \Theta^3_1&\Theta^3_2&\Theta^3_3
    \end{pmatrix},
\end{align}
where $\bv{\Theta}=(\bv{\Theta}^1, \bv{\Theta}^2, \boldsymbol{\Theta}^3)^T$.
$\bv{\Theta}$ is referred to as the plastic displacement gradient because it corresponds to the displacement gradient $\nabla \bv{u}$ in conventional elasticity theory.
Evidently, $\bv{\Theta}$ primarily contributes to the plastic deformation gradient.
As will be elaborated in the following section, it plays a crucial role in the present geometrical theory of dislocations.
For the exterior derivative $d$ and codifferential $-* d*$, this results in Cartan's first structure equation $d\bv{\Theta}^i=\bv{\tau}^i$ along with the divergence-free condition $-*d* \bv{\Theta}^i=0$ ensuring that $\bv{\Theta}^i$ is the dual exact form \cite{wenzelburger_kinematic_1998, kobayashi_geometrical_2024}.
The divergence-free condition provides a new constraint on $\bv{\Theta}^i$ that results in a unique solution to Cartan's equation.
Furthermore, the exterior and codifferential of the 1-form $\bv{\Theta}^i$ can be expressed as the curl and divergence of a vector-valued function, respectively \cite{schwarz_hodge_1995}.
Consequently, these differential equations can be reformulated using vector analysis as follows:
\begin{align}
  \label{eq:Cartan}
  \nabla \times \bv{\Theta}^i=\bv{\alpha}^i,
  \quad
  \nabla \cdot \bv{\Theta}^i=0,
\end{align}
where $\nabla\times$ and $\nabla \cdot $ represent the curl and divergence operators of $\mathbb{R}^3$.
The same governing equations linking a dislocation density tensor and plastic deformation were derived in the field theory of dislocations by utilising Helmholtz decomposition \cite{acharya_model_2001,arora_finite_2020,acharya_structure_2019}.
After solving the aforementioned equation, the Riemannian metric can be determined as follows:
\begin{align}
\label{eq:IntermediateMetric}
    \bv{g}_\mathcal{B}=\delta_{ij}\bv{\vartheta}^i\otimes \bv{\vartheta}^j=
    \delta_{ij}(F_p)^i_k (F_p)^j_l d\bv{x}^k\otimes d\bv{x}^l.
\end{align}
Consequently, the intermediate configuration $\mathcal{B}$ can be promptly obtained after integrating equation (\ref{eq:Cartan}).

\subsection{Elastic embedding to the current configuration}
In the geometric dislocation theory, the second step involves determining the current configuration.
This process entails identifying an embedding of the intermediate configuration $\mathcal{B}$ into the Euclidean space $\mathbb{R}^3$, which corresponds to elastic deformation in mechanics.
Consequently, the relationship between the dual bases of the intermediate and current configurations can be expressed as follows: $d\bv{y}^i=(F_e)^i_j \bv{\vartheta}^j=(F_t)^i_j d\bv{x}^j$.
where $\bv{F}_t$ and $\bv{F}_e$ denote the total and elastic components of the deformation gradient, respectively.
According to multiplicative decomposition, we have $\bv{F}_t=\bv{F}_e \bv{F}_p=\pd{\bv{y}}{\bv{x}}=\bv{I}+\nabla\bv{u}$, where $\nabla\bv{u}$ represents the displacement gradient and $\bv{I}$ denotes the $3\times 3$ identity matrix.
The Riemannian metrics of the current configurations are defined as follows:
\begin{align}
    \bv{g}_\mathcal{C}=\delta_{ij}d\bv{y}^i\otimes d\bv{y}^j
    =\delta_{ij} (F_t)^i_k (F_t)^j_l d\bv{x}^k\otimes d\bv{x}^l.
\end{align}

Following the standard theories of geometrical elasticity, the elastic Green--Lagrange strain $\bv{E}$ is characterised by the difference in the Riemannian metrics between the two configurations: $\bv{E} = (\bv{g}_\mathcal{C} - \bv{g}_\mathcal{B})/2$ \cite{marsden_mathematical_1984}.
Subsequently, the strain energy density of the St. Venant--Kirchhoff hyperelasticity and the second Piola--Kirchhoff stress tensor $\bv{S}$ was introduced \cite{grubic_equations_2014}:
\begin{align}\label{eq:strainenergy}
  \mathcal{W}=\frac{1}{2} \bv{E}:\bv{C}:\bv{E}\sqrt{\det \bv{g}_\mathcal{B}},\quad \bv{S} = \pd{\mathcal{W}}{\bv{E}}=\bv{C}:\bv{E}\sqrt{\det \bv{g}_\mathcal{B}},
\end{align}
where $\bv{C}$ denotes the 4th-rank elastic stiffness tensor.
For an isotropic medium with the Poisson's ratio $\nu$ and shear modulus $\mu$, the components are denoted using the inverse of the metric tensor $g^{ij}_\mathcal{B}$ as $C^{ijkl}=\frac{2\nu\mu}{1-2\nu} g^{ij}_\mathcal{B} g^{kl}_\mathcal{B} + \mu(g^{ik}_\mathcal{B} g^{jl}_\mathcal{B}+g^{il}_\mathcal{B} g^{jk}_\mathcal{B})$.
Given a specific plastic displacement gradient $\bv{\Theta}^i$, the embedding $\bv{y}$ was determined by solving the following stress equilibrium equation \cite{grubic_equations_2014}
\begin{align}\label{eq:stressequil}
  \nabla\cdot \bv{P}=0.
\end{align}
where $\bv{P}=\bv{F} \cdot \bv{S}$ represents the first Piola--Kirchhoff stress tensor.
This is known as the stress equilibrium equation.

\section{Analytical integration for Cartan's equation}

\subsection{Biot--Savart law: electromagnetics and dislocations}
\label{sec: biotsavartlaw}

As elucidated in the preceding section, the fundamental concept of the geometrical theory of dislocation involves the introduction of an intermediate configuration represented as a Riemann--Cartan manifold.
The mathematical formulation of the intermediate configuration is established by solving equation (\ref{eq:Cartan}) for a given dislocation distribution $\bv{\alpha}$.
The subsequent analysis of elastic deformation involves embedding the intermediate configuration into Euclidean space $\mathbb{R}^3$.
Despite incorporating a geometric process, the transition of the Riemannian metric from the intermediate to the reference configuration results in an analysis that closely resembles that of nonlinear elasticity theory.
Hence, the primary challenge in geometrical analysis lies in the construction of the intermediate configuration, which is performed by utilising Cartan's equation (\ref{eq:Cartan})${}_1$ and divergence-free condition (\ref{eq:Cartan})${}_2$.
In a previous study, we proposed a numerical method for an arbitrary configuration of dislocations \cite{kobayashi_geometrical_2024}, although an analytical solution for such configurations has not yet been reported.
The present focus shifts to the mathematical properties of the mathematical features of equation (\ref{eq:Cartan}).
Because these are linear partial differential equations, analytical solutions $\bv{\Theta}^i$ can be anticipated.
Cartan's equation further suggests that the presence of non-vanishing torsion results in the formation of a circular field in $\bv{\Theta}^i$ that is free from divergence.
Surprisingly, a field of physics exists in which the governing equations are precisely akin to this mathematical system: electromagnetism.

Static magnetic flux density fields induced by steady-state electric current are treated in magnetostatics.
This physical field is well-established and widely recognised; however, a brief overview is provided here for comparison with the geometric dislocation theory.
Let $\bv{x}=(x,y,z)$ and $\bv{\xi}=(\xi, \eta, \zeta)$ represents points in $\mathbb{R}^3$, and $\bv{B}=\bv{B}(\bv{x})$ denotes the magnetic flux density induced by the electric current density $\bv{J}=\bv{J}(\bv{x})$.
We assume that $\bv{B}$ and $\bv{J}$ are in a steady state, meaning they do not vary with time.
Therefore, the differential form of Amp\`ere's law and Gauss' law for the magnetic flux density $\bv{B}$ are expressed as follows:
\begin{align}
  \label{eq:gauss-ampere}
    \nabla\times \bv{B}= \mu_0\bv{J},\quad
    \nabla\cdot \bv{B}=0.
\end{align}
The analytical solutions to these equations are provided by the following theorem \cite{griffiths_introduction_2023}:

\begin{theorem}[Biot--Savart law of electromagnetics]\label{th:BS-EM}
Consider a steady-state electric current $\bv{J}(\bv{\xi})$ in a vacuum.
The resulting steady-state magnetic flux density $\bv{B}(\bv{x})$ surrounding the electric current is expressed as follows:
\begin{align}\label{eq:BS}
  \bv{B}(\bv{x}) = 
  \frac{\mu_0}{4\pi} \int_{\mathbb{R}^3}
  \frac{\bv{J}(\bv{\xi}) \times (\bv{x}-\bv{\xi})}{\|\bv{x}-\bv{\xi}\|^3} dV,
\end{align}
where $\mu_0$ represents the permeability of the vacuum, and integration is performed for $\bv{\xi}$ over $\mathbb{R}^3$.
Moreover, let $\bv{A}(\bv{x})$ denote the vector potential defined by
\begin{align}\label{eq:VectorPotential}
  \bv{A}(\bv{x}) = 
  \frac{\mu_0}{4\pi} \int_{\mathbb{R}^3}
  \frac{\bv{J}(\bv{\xi}) }{\|\bv{x}-\bv{\xi}\|} dV.
\end{align}
The magnetic flux density is then derived from this potential, such that $\bv{B} = \nabla \times \bv{A}$.
\end{theorem}

\begin{proof}
The theorem is validated by substituting equations (\ref{eq:BS}) and (\ref{eq:VectorPotential}) into equation (\ref{eq:gauss-ampere}).
Further details can be obtained from \cite{griffiths_introduction_2023}.
\end{proof}

The Biot--Savart law is derived from the vector potential $\bv{A}(\bv{x})$ that is determined by the current distribution $\bv{J}(\bv{\xi})$.
By taking the curl of the vector potential and utilising the properties of the curl and the gradient of the inverse distance, the expression of the magnetic flux density $\bv{B}$ owing to the current-carrying wire can be derived, validating the Biot--Savart law.
Notably, the wire does not need to be a straight line.
The mathematical structure of magnetostatics, as presented in equation (\ref{eq:gauss-ampere}) is identical to that of the geometrical theory of dislocations expressed in equation (\ref{eq:Cartan}).
This similarity suggests that the theory of dislocations in the Riemann--Cartan manifold is analogous to that of steady-state electromagnetics.
The analytical solution to equation (\ref{eq:Cartan}) is obtained using the Biot--Savart law presented in equation (\ref{eq:BS}).
This hypothesis was validated through direct calculations (see equations (\ref{eq:bs cartan}) and (\ref{eq:Cartan})).
These findings can be summarised by utilising the following theorem.

\begin{theorem}[Biot--Savart law of dislocation]\label{th:BS}
Let $\bv{\tau}^i(\bv{\xi})=\ast \bv{\alpha}^i(\bv{\xi})$ be the torsion 2-form resulting from dislocations in an infinite medium.
Then, the analytical integration of the dual exact form of the plastic displacement gradient $\bv{\Theta}^i(\bv{x})$ can be expressed as follows:
\begin{align}
  \label{eq:bs cartan}
   \bv{\Theta}^i(\bv{x}) =
  \frac{1}{4\pi}
  \int_{\mathbb{R}^3}
  \frac{\bv{\alpha}^i(\bv{\bv{\xi}})\times (\bv{x}-\bv{\xi})}{\|\bv{x}-\bv{\xi}\|^3} dV,
\end{align}
for $i=1,2,3$.
Volume integration is performed for $\bv{\xi}$ over $\mathbb{R}^3$.
Moreover, the vector potential of the plastic displacement gradient is obtained as follows:
\begin{align}
\label{eq:VectorPotential_cartan}
   \bv{\Psi}^i(\bv{x}) =
  \frac{1}{4\pi}
  \int_{\mathbb{R}^3}
  \frac{\bv{\alpha}^i(\bv{\bv{\xi}})}{\|\bv{x}-\bv{\xi}\|} dV,
\end{align}
satisfying $\bv{\Theta}^i=\nabla \times \bv{\Psi}^i$.
\end{theorem}

\begin{proof}
  While the dislocation density $\bv{\alpha}(\bv{\xi})$ is an $\mathbb{R}^3$-valued 1-form, the electric current $\bv{J}(\bv{\xi})$ in theorem \ref{th:BS-EM} is a 1-form which takes the value in $\mathbb{R}$.
  A similar relationship exists between the magnetic flux density $\bv{B}$ and $\bv{\Theta}$.
  In addition, the linearity of equation (\ref{eq:Cartan}) allows for the complete decoupling of the components of the fields $\bv{\alpha}^i$ and $\bv{\Theta}^i$ ($i=1,2,3$).
  Therefore, by replacing $\bv{J}$ and $\bv{B}$ with $\bv{\alpha}^i$ and $\bv{\Theta}^i$, and substituting $\mu_0=1$, Cartan's equation (\ref{eq:Cartan}) for fixed $i$ can be recovered from the governing equations in magnetostatics (\ref{eq:gauss-ampere}).
  The direct application of Theorem \ref{th:BS-EM} yields equations (\ref{eq:bs cartan}), (\ref{eq:VectorPotential_cartan}), and $\bv{\Theta}^i=\nabla\times\bv{\Psi}^i$.
\end{proof}

Notably, theorem \ref{th:BS} is applicable not only to a single dislocation but also to multiple dislocations, including edge and screw dislocations, whose distribution is characterised by $\bv{\tau}^i =*\bv{\alpha}$.
The superposition of solutions is made possible by the linearity of equation (\ref{eq:Cartan}).
The boundary conditions can also be addressed using the method of images \cite{griffiths_introduction_2023}.
By leveraging the Biot--Savart law of dislocations and Helmholtz decomposition, the plastic deformation gradient $\bv{F}_p$ can be obtained for an arbitrary configuration of dislocations, leading to the analytical expression for the intermediate configuration $\mathcal{B}$ of the Riemann--Cartan manifold.
Previous studies have utilised various methods for analysing Cartan's first structure equation, such as the Homotopy operator \cite{acharya_model_2001,edelen_gauge_1988}, the semi-inverse method \cite{yavari_riemann--cartan_2012,yavari_geometry_2014,clayton_defects_2015}, and variational methods \cite{arora_finite_2020,kobayashi_geometrical_2024}.
Compared with these methods, the present Biot--Savart law offers several advantages.
One significant benefit is its wide applicability; this method is suitable for arbitrary configurations of dislocations, whereas the Homotopy operator and semi-inverse method require high spatial symmetry and cannot be applied to arbitrary dislocation configurations.
In addition, the variational method utilised in previous studies is not suitable for mathematical analysis because it provides only numerical solutions.
As discussed in subsequent sections, the derivation of analytical solutions for the intermediate configuration presents a significant intersection with other scientific disciplines, particularly electromagnetism and complex function theory, which were previously thought to be unrelated to dislocation theory.
This interdisciplinary connection is anticipated as crucial in the fundamental development of the dislocation theory.

\subsection{Plastic deformation fields around a straight screw dislocation}

The theory formulated in the previous section provides a rigorous mathematical representation of the intermediate configuration $\mathcal{B}$ corresponding to an arbitrary dislocation distribution.
Here, we demonstrate the applicability of the Biot--Savart law to the straight screw dislocation, culminating in the derivation of the analytical expression for the plastic deformation gradient $\bv{F}_p$ associated with the screw dislocation.

A schematic of the analytical model of the straight screw dislocation is shown in Figure \ref{fig.1}(a).
The screw dislocation is located along the $z$-axis.
Hence, the Burgers vector is defined as $\bv{b} = (0, 0, b)$, and the tangent of the dislocation line $\bv{n} = (0, 0, 1)$ is localised along the $z$-axis.
By leveraging the Biot--Savart law, we analytically determined the plastic deformation gradient $\bv{F}_p$, as summarised in the ensuing theorem:

\begin{theorem}[Plastic deformation fields of screw dislocation]\label{th:ScrewDislocation}
Let $\bv{b} = (0, 0, b)$ be the Burgers vector of the screw dislocation whose dislocation line is localised along the $z$-axis, \textit{i.e.}, $\bv{n} = (0, 0, 1)$ (see figure \ref{fig.1}(a)).
Then, the analytical expression for the plastic deformation gradient $\bv{F}_p^{S}$ is given by the following form:
\begin{align}
  \label{eq:Fp_screw}
  \bv{F}_p^S= \bv{I}+\frac{b}{2\pi}\left(
  \begin{array}{ccc}
  0&0&0\\
  0&0&0\\
  -\dfrac{y}{x^2+y^2}&\dfrac{x}{x^2+y^2}&0\\
  \end{array}\right).
\end{align}
\end{theorem}

\begin{proof}
  As shown in figure \ref{fig.1}(a), the screw dislocation possess a dislocation line $\bv{n}=(0,0,1)$ with the Burgers vector $\bv{b}=(0,0,b)$ localised along the dislocation line $(x,y)=(0,0)$.
  By definition, the spatial distribution of the dislocation density becomes $\bv{\alpha}^1=\bv{\alpha}^2=\bv{0}$ and
  \begin{align}
    \label{eq:ScrewDislocationDensity}
    \bv{\alpha}^3=\alpha^3_1dx+\alpha^3_2 dy+\alpha^3_3 dz
    =b\delta(x)\delta(y)dz,
  \end{align}
  where $\delta$ is the Dirac delta.
  By incorporating the dislocation density $\bv{\alpha}^i$ into the Biot--Savart law of dislocations (\ref{eq:bs cartan}), we obtain the analytical expression for the plastic displacement gradient $\bv{\Theta}^1=\bv{\Theta}^2=\bv{0}$ and 
  \begin{align}\label{eq:PlasticGradientScrew}
    \bv{\Theta}^3(\bv{x})
    = \frac{b}{2\pi} \left( -\frac{y}{x^2+y^2}, \frac{x}{x^2+y^2}, 0\right).
  \end{align}
  From equation (\ref{eq:Helmholtz}), we obtain $\bv{F}_p^S=\bv{I} + \bv{\Theta}^S$ where we set $\bv{\Theta}^S=(\bv{\Theta}^1, \bv{\Theta}^2, \bv{\Theta}^3)^T$.
  A direct calculation validates that the plastic displacement gradient of the screw dislocation $\bv{\Theta}$ satisfies both Cartan's first structure equation (\ref{eq:Cartan})$_1$ and the divergence-free condition (\ref{eq:Cartan})$_2$.
  This proves the Theorem \ref{th:ScrewDislocation}.
\end{proof}

A fundamental concept in understanding Amp\`ere's law is the right-hand screw rule, which explains the formation of a vortex-like magnetic flux density when a linear steady-state current aligns with the direction of the right-hand screw \cite{griffiths_introduction_2023}.
As shown in the previous section, the analogy between the steady-state electromagnetic field theory described by Amp\`ere's and Gauss' laws and the plastic displacement gradient of dislocations given by Cartan's first structure equation is well-established.
The plastic displacement gradient component $\bv{\Theta}^3$ created by the screw dislocation exhibits a similar vortex-like structure, as shown in the lower section of figure \ref{fig.1} (a).
A counterclockwise rotation centred at the dislocation line was observed on a cross section perpendicular to the dislocation line.
The decrease in the vector magnitude was inversely proportional to the square of the distance from the origin.
Additionally, the dual bases $\bv{\vartheta}^3$ are indicated as vector fields in the middle section in figure \ref{fig.1}(b).
The direction of the vectors aligned with the $z$-axis, distinguishing them from the dual coordinate basis $dz$ when the counterclockwise vortex centred at the axis was superposed.
Consequently, the integration curve initiated from the section formed a part of a helix.
This helical curve represents the outline of the intermediate configuration that cannot be fully represented in the Euclidean space.
The plastic displacement gradient $\bv{\Theta}$ emerged as a crucial factor in the intermediate configuration, as evident from Cartan's equation and the Helmholtz decomposition.

\begin{figure}[!h]
  \centering
    \includegraphics[width=0.95\textwidth]{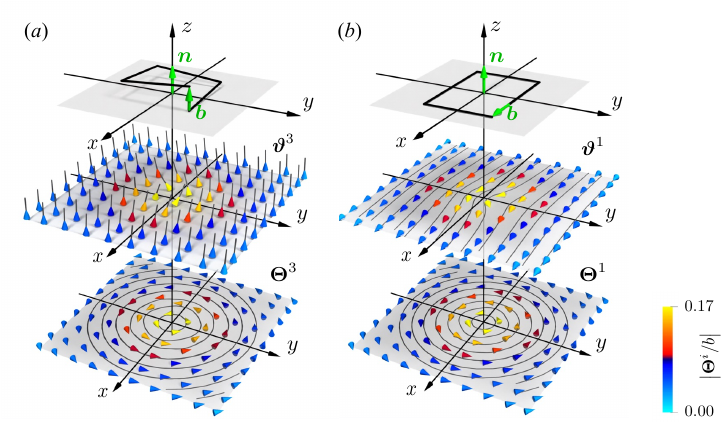}
  \caption{
    Schematic of (a) screw and (b) edge dislocation, and its plastic deformation.
    In both (a) and (b), the top, middle, and bottom sections indicate the dislocation model, dual basis $\bv{\vartheta}^i$, and plastic displacement gradient $\bv{\Theta}^i$, respectively.
    The bottom section highlights that $\bv{\Theta}^3$ and $\bv{\Theta}^1$ of the screw and edge dislocations exhibit the same counterclockwise vortex vector field, which diverges as the distance from the dislocation line increases.
    The dual bases $\bv{\vartheta}^3$ of the screw and $\bv{\vartheta}^1$ of the edge dislocations differ; $\bv{\vartheta}^3$ points outward from the section, whereas $\bv{\vartheta}^1$ points inward, which is directly indicated by the integral curve of the vector field illustrated as the dark curves on the each planes.
  }
  \label{fig.1}
\end{figure}

\subsection{Plastic deformation fields around a straight edge dislocation}
In the preceding section, we derived the plastic deformation gradient of straight screw dislocations.
The Biot--Savart law is universally applicable regardless of the type of dislocations.
In this section, we will apply it to the straight edge dislocation and derive the corresponding plastic deformation gradient.

The edge dislocation model is shown in figure \ref{fig.1}(b).
The Burgers vector is $\bv{b} = (b, 0, 0)$, and the tangent of the dislocation line $\bv{n} = (0, 0, 1)$ is localised along the $z$-axis.
By utilising the Biot--Savart law, the plastic deformation gradient of the edge dislocation can be analytically determined, which is summarised in the following theorem:

\begin{theorem}[Plastic deformation fields of edge dislocation]\label{th:EdgeDislocation}
Let $\bv{b} = (b, 0, 0)$ be the Burgers vector of the edge dislocation with its dislocation line lines on the $z$-axis, \textit{i.e.}, $\bv{n} = (0, 0, 1)$ (figure \ref{fig.1}(b)).
The analytical expression of the plastic deformation gradient $\bv{F}_p^{E}$ is expressed as follows:
\begin{align}
  \label{eq:Fp_edge}
  \bv{F}_p^E= \bv{I}+\frac{b}{2\pi}\left(
  \begin{array}{ccc}
  -\dfrac{y}{x^2+y^2}&\dfrac{x}{x^2+y^2}&0\\
  0&0&0\\
  0&0&0\\
  \end{array}\right).
\end{align}
\end{theorem}

\begin{proof}
  As shown in figure \ref{fig.1}(b), the edge dislocation has the Burgers vector $\bv{b}=(b,0,0)$ and the dislocation line $\bv{n}=(0,0,1)$ localised along the $z$-axis.
  The spatial distribution of the dislocation density becomes $\bv{\alpha}^2=\bv{\alpha}^3=\bv{0}$ and
  \begin{align}
      \bv{\alpha}^1=\alpha^1_1dx+\alpha^1_2 dy+\alpha^1_3 dz
      =b\delta(x)\delta(y)dz.
  \end{align}
  By substituting the dislocation density $\bv{\alpha}^i$ into the Biot--Savart law of dislocation (\ref{eq:bs cartan}), the analytical expression of the plastic displacement gradient becomes $\bv{\Theta}^2=\bv{\Theta}^3=\bv{0}$ and
  \begin{align}\label{eq:PlasticGradientEdge}
      \bv{\Theta}^1 = \frac{b}{2\pi} \left( -\frac{y}{x^2+y^2}, \frac{x}{x^2+y^2}, 0\right).
  \end{align}
  From equation (\ref{eq:Helmholtz}), we obtain $\bv{F}_p^E=\bv{I} + \bv{\Theta}^E$, where we set $\bv{\Theta}^E=(\bv{\Theta}^1, \bv{\Theta}^2, \bv{\Theta}^3)^T$.
  A direct calculation validated that the plastic deformation gradient of the edge dislocation $\bv{F}_p^E$ satisfied Cartan's first structure equation (\ref{eq:Cartan})$_1$ and divergence-free condition (\ref{eq:Cartan})$_2$, simultaneously, validating the theorem \ref{th:EdgeDislocation}.
\end{proof}

The plastic deformation of edge dislocations is shown in the bottom section of figure \ref{fig.1}(b).
In this analysis, the dislocation lines are localised at the origin of the $xy$-plane, resulting in a singularity in the plastic displacement field.
By excluding this singular point from the analysis, the plastic displacement gradient $\bv{\Theta}^1$ formed a counterclockwise vortex on a cross section perpendicular to the dislocation line, with its magnitude decreasing inversely proportional to the distance and approaching zero at infinity.
A key finding of this study is that this vortex structure is consistent across all types of dislocations.
The vector field plot of the dual basis $\bv{\vartheta}^1$ is shown in the middle of figure \ref{fig.1}(b).
Unlike the screw dislocation, the vector field was distributed within the section, and the vortex field was not as prominent.
Instead, the magnitude of the vector field was less than 1 for $y > 0$ and greater than 1 for $y < 0$ along the $x$-axis, indicating stretching and shrinking of the line element compared with that of the reference configuration.
This trend contrasts with that of a screw dislocation, which always results in stretching.
The stretching and shrinking of the dual frame likely contributed significantly to the formation of the in-plane stress field of the edge dislocation.

While the plastic deformation gradient of screw and edge dislocations differs, their plastic displacement gradients align perfectly, as demonstrated in Theorems \ref{th:ScrewDislocation} and \ref{th:EdgeDislocation}.
As discussed in the following section, the elastic mechanical fields of screw and edge dislocations differ considerably.
The distinction between them lies in the component to which the identity matrix is added, although the identity matrix itself does not provide information on the plastic deformation of the dislocation.
Therefore, the essence of plastic deformation owing to dislocations lies in the plastic displacement gradient, which is independent of the type of dislocation.
However, the plastic displacement gradients of screw and edge dislocations align perfectly owing to essentially identical dislocation density in constructing the Riemann--Cartan manifold for these two lattice defects.
The Biot--Savart law can be utilised to analytically determine the plastic deformation gradient for a straight dislocation arrangement.

\subsection{Relevance to the classical dislocation theory}

Here, we now focus on the theoretical analysis of plastic deformation.
Notably, the geometrical and conventional theories of dislocation differ significantly.
In electromagnetism, the magnetic flux density formed around a steady-state electric current is derived using the Biot--Savart law, which relies on two fundamental equations: Amp\`ere's and Gauss' law.
Similarly, in the geometrical dislocation theory, the plastic displacement gradient field formed around the dislocation density is described by two governing equations, Cartan's equation and the divergence-free condition. However, the mathematical similarity of these equations has led to the discovery of a previously unknown Biot--Savart law of plastic deformation.
Based on the conventional dislocation theory, the relationship between dislocation density and plastic displacement gradient, which corresponds to Cartan's structural equation, has already been established \cite{kroner_allgemeine_1959,dewit_theory_1973,dewit_theory_1973-1,de_wit_view_1981,lazar_dislocations_2006}.
However, the relationship corresponding to the divergence-free condition was explored only in a few studies \cite{wenzelburger_kinematic_1998,acharya_model_2001,roy_finite_2005,acharya_structure_2019}, and therefore, the analogy with electromagnetics was not fully be considered.
As previously mentioned, the divergence-free condition resulted from the Helmholtz decomposition of the plastic deformation gradient, which mathematically originated from the fact that $\bv{\vartheta}$ assumes the $\R{3}$-valued 1-form.
Moreover, recent studies has demonstrated the duality between conventional elasticity theory of defects in two-dimensional crystals and fractional excitations in quantum systems \cite{pretko_fracton-elasticity_2018,pretko_symmetry-enriched_2018,pretko_crystal--fracton_2019,gromov_colloquium_2024}.
This duality is corroborated by the tensor gauge theory, which results in the governing equations akin to the that of the electromagnetics.
This aspect also suggests the validity of the analogy between magnetostatics and the geometrical theory of dislocations.

\section{Discussion}

\subsection{Cauchy--Riemann equation and plastic deformation potential}

The discussion thus far indicates that the plastic deformation field surrounding the dislocations mirrors the distribution of a static magnetic flux density field around an electric current.
This mathematical alignment is underpinned by the orthogonality of $\bv{\Theta}^i$ and $\bv{\alpha}^i$.
As a stationary current generates a static magnetic field, the dislocation density tensor induces plastic deformation in the direction orthogonal to the dislocation line.
Apart from this orthogonality, another layer of orthogonality in the plastic deformation fields is introduced here, that is the conformal property.
This additional orthogonality of the plastic deformation fields reveals another remarkable mathematical characteristic of the geometrical theory of dislocations.
This aspect pertains to its connection with the theory of functions of a complex variable.
Particularly, the conformal property of the plastic deformation gradient serves as the foundation for deriving a complex functional representation of the plastic deformation fields.
Hereafter, we denote a complex number as $w=x+\mathrm{i} y$.
In standard complex analysis, we examine two smooth functions, $u=u(w)$ and $v=v(w)$ defined on the complex plane $\mathbb{C}$ (or $xy$-plane).
The complex function $f(w)=u(w)+\mathrm{i} v(w)$ is conformal on a domain in the complex plane if and only if its complex derivative exist on it \cite{arfken_mathematical_2013}.
Notably, complex differentiability is equivalent to the Cauchy--Riemann equation.
By applying this framework, an analytical representation using complex functions can be obtained for the plastic deformation field of a straight dislocation as discussed in the previous section.
We reiterated that the plastic deformation field of a dislocation represents a conformal map.

\begin{theorem}[Conformal property of plastic deformation fields of dislocations]
  \label{thm:conformity}
  Let $\bv{b} = (0, 0, b)$ be Burgers vector of a screw dislocation whose dislocation line lies along the $z$-axis, \textit{i.e.}, $\bv{n} = (0, 0, 1)$ (see figure \ref{fig.1}(a)), and let $\bv{\alpha}^3$ be the corresponding dislocation density tensor.
  The plastic displacement gradient $\bv{\Theta}^3$ on the $xy$-plane, excluding the origin $(x,y)=(0,0)$, represents a conformal map.
  A similar concept applies to the plastic displacement fields of edge dislocations.
\end{theorem}

\begin{proof}
  According to equation (\ref{eq:ScrewDislocationDensity}), the dislocation density of the screw dislocation is expressed as $\bv{\alpha}^3=b \delta(x)\delta (y)dz$.
  In contrast, from equation (\ref{eq:Fp_screw}), the plastic displacement fields $\bv{\Theta}^3$ are distributed on any $xy$-plane, such that $\bv{\Theta}^3(x,y)=(\Theta^3_1(x,y), \Theta^3_2(x,y))$, whereas its distribution is uniform along the $z$-direction.
  The dislocation density and plastic displacement gradient satisfy Cartan's first structure equation (\ref{eq:Cartan}): $\nabla\times \bv{\Theta}^3=b \delta(x)\delta (y)$.
  Obviously, the dislocation density $\bv{\alpha}^3$ is localised along the dislocation line and therefore only exists at the origin $(0,0)$.
  Hence, Cartan's equation and the divergence-free condition hold true almost everywhere except at the origin $(0,0)$ such that
  \begin{align}
    \label{eq:cauchy-riemann}
    \pd{\Theta^3_1}{y}-\pd{\Theta^3_2}{x}=0,\quad
    \pd{\Theta^3_1}{x}+\pd{\Theta^3_2}{y}=0.
  \end{align}
  We introduce the complex function defined as
  \begin{align}\label{eq:ComplexFunction}
      f(x+\mathrm{i} y)=\Theta^3_1(x,y)-\mathrm{i} \Theta_2^3(x,y).
  \end{align}
  In other words, $\Theta^3_1$ and $-\Theta^3_2$ represent the real and imaginary parts of the complex function $f$.
  Therefore, equation (\ref{eq:cauchy-riemann}) is essentially the Cauchy--Riemann equation for the complex function $f$.
  These equations are also equivalent to the differential of the function $f$ with respect to the complex conjugate $\bar{w}=x-\mathrm{i}y$ such that $\partial f/\partial \bar{w}=(\partial f/ \partial x+\mathrm{i}\partial f/\partial y)/2=0$.
  This result indicates that the function $f$ maintains the orthogonality of the original $(x,y)$-coordinates following the mapping onto the complex plane.
  In other words, the plastic deformation field $(\Theta_1^3, \Theta_2^3)$ defines an orthogonal coordinate system on the complex plane, excluding the origin $(0,0)$.
\end{proof}

The plastic displacement gradient $\bv{\Theta}^i$ resulting from a dislocation exhibits two types of orthogonality.
First, they are distributed in a plane orthogonal to the dislocation line.
This orthogonality allows the rows of the plastic displacement gradient to be decoupled and independently determined for each coordinate component of the Burgers vector in Cartan's equation (\ref{eq:Cartan}).
Consequently, the components of the Burgers vector along the three coordinate axes possess independent dislocation densities, as shown in equation (\ref{eq:DislocationDensity-Torsion}).
Hence, they do not interact with each other and can be superimposed.
The second type of orthogonality pertains to the in-plane $\Theta^i$-components.
This implies that the plastic displacement gradient functions act as a conformal map, resulting in a circular magnetic field similar to that described by Amp\`ere's law.

This leads to the following outcomes, which are crucial in the analysis of the plastic deformation of dislocations using complex functions.

\begin{theorem}[Plastic deformation potential]
\label{thm:PlasticPotential}
Let $\Psi(w)$ represent a complex function in the form
\begin{align}
\label{eq:PlasticPotential}
\Psi(w) = -\mathrm{i} \frac{b}{2\pi} \log w.
\end{align}
Then, the plastic displacement gradients of the screw dislocations are obtained from the complex derivative of $\Psi$ as follows:
\begin{align}
\Theta^3_1 = \mathrm{Re}\left( \frac{d \Psi}{d w} \right),
\quad
\Theta^3_2 = -\mathrm{Im}\left( \frac{d \Psi}{d w} \right).
\end{align}
The same holds true for edge dislocations: $\Theta^1_1 = \mathrm{Re} \left( d\Psi/dw \right)$ and $\Theta^1_2 = -\mathrm{Im} \left( d\Psi/dw \right)$.
The function $\Psi(w)$ is referred to as the complex potential of plastic deformation.
\end{theorem}

\begin{proof}
  The differential operator with respect to the complex variable $w$ is defined as $d/dw=(\partial/\partial x-\mathrm{i}\partial/\partial y)/2$.
  The partial derivatives of the potential are denoted as $\partial \Psi/\partial x = \mathrm{i}b/2\pi w$ and $\partial \Psi/\partial y = -b/2\pi w$.
  Hence, from a direct calculation, we obtain
  \begin{align}
  \frac{d\Psi}{dw}
  =\frac{1}{2}\left(
  \frac{\partial }{\partial x}-\mathrm{i}\frac{\partial}{\partial y}
  \right)\Psi
  =-\frac{b}{2\pi} \frac{y-\mathrm{i}x}{x^2+y^2}.
  \end{align}
  Clearly, $\mathrm{Re}(d\Psi/dw)$ and $-\mathrm{Im}(d\Psi/dw)$ correspond to $\Theta^3_1$ and $\Theta^3_2$ of the screw dislocation, respectively.
  The plastic displacement gradient $\bv{\Theta}^3$ of the screw dislocation aligns with $\bv{\Theta}^1$ of the edge dislocation, indicating that the function (\ref{eq:PlasticPotential}) represents the plastic potential of the edge dislocations.
\end{proof}

The properties of the complex potential of plastic deformation (\ref{eq:PlasticPotential}) can be better understood by changing the representation from Cartesian $(x,y)$-coordinates in the complex plane to polar $(r,\theta)$-coordinates.
With the polar representation $w=r e^{\mathrm{i}\theta}$, the complex potential given in equation (\ref{eq:PlasticPotential}) is written in the polar coordinate system in such a way that
\begin{align}
    \Psi(r, \theta)=-\frac{b}{2\pi}(\theta+\mathrm{i}\log r).
\end{align}
The real and imaginary parts of the complex function in the polar coordinate system are shown in figures \ref{fig.2}.
The real part $\mathrm{Re}(\Psi)=-\frac{b}{2\pi}\theta$ forms a helical surface with coefficient $-b/2\pi$.
This result indicates that the potential is multivalued, with a singularity present at the axial origin, \textit{i.e.}, at the dislocation centre.
Mathematically, the dislocation centre $z=0$ is the branch point.
However, the imaginary part $\mathrm{Im}(\Psi)=-\frac{b}{2\pi}\log r$ exhibits a logarithmic decay as the distance from the origin increases.
The analysis of the displacement field generated by dislocations is considered as topological lattice defects involving topological singularities \cite{acharya_structure_2019}.
While this understanding is accurate, this study has successfully extract the plastic deformation field from the total deformation field of the lattice defect.
This extraction validated that the topological change resulted from plastic deformation exclusively, rather than being directly linked to elastic deformation.

\begin{figure}[!h]
  \centering
  \includegraphics[width=0.95\textwidth]{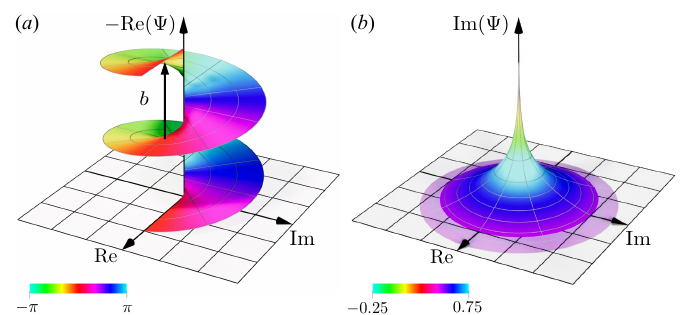}
  \caption{
    Visualisation results of the plastic deformation potential owing to dislocations. (a) Real part of the potential shows a multi-valued helical surface, representing the topological property of dislocation. (b) Imaginary part of the potential shows logarithmic decay.
  }
  \label{fig.2}
\end{figure}

\subsection{Linearised stress fields around dislocations}
The geometrical theory of dislocations is characterised by its unique approach to the treatment of the intermediate configuration.
Specifically, plastic deformation is not represented within the usual Euclidean space but is instead modelled using mathematically generalised Riemann--Cartan manifolds.
The current configuration $\mathcal{C}$ is achieved by embedding the intermediate configuration into Euclidean space.
The embedding map represents elastic deformation, which is determined by solving the stress equilibrium equation (\ref{eq:stressequil}).
Owing to the non-linear nature of this partial differential equation, obtaining an analytical solution is typically impractical.
However, a linear approximation enables to determine the stress field around a dislocation analytically.
Our proposed theory can be objectively validated by comparing the stress fields obtained through the geometrical approach with those utilised in Volterra dislocation theory.

First, we consider the linearisation of the kinematics.
The linearised Cauchy strains for total deformation $\bv{\mathcal{E}}_t$, plastic deformation $\bv{\mathcal{E}}_p$ and elastic deformation $\bv{\mathcal{E}}_e$ are expressed as follows:
\begin{align}\label{eq:LinearStrains}
  \bv{\mathcal{E}}_t={}\frac{1}{2}(\bv{\nabla u}+\bv{\nabla u}^T),
  \quad
  \bv{\mathcal{E}}_p={}\frac12(\bv{\Theta}+\bv{\Theta}^T),
  \quad
  \bv{\mathcal{E}}_e=\bv{\mathcal{E}}_t-\bv{\mathcal{E}}_p.
\end{align}
The elastic stress $\bv{\sigma}$ is determined using the linearised Hooke's law, such that $\bv{\sigma}=\bv{C}_0:\bv{\mathcal{E}}_e$ where $\bv{C}_0$ denotes the standard elastic coefficients defined by $C^{ijkl}_0=\frac{2\nu\mu}{1-2\nu} \delta^{ij} \delta^{kl} + \mu(\delta^{ik}\delta^{jl}+\delta^{il}\delta^{jk})$ \cite{marsden_mathematical_1984}.
Based on Hooke's law, the plastic Cauchy strain $\bv{\mathcal{E}}_p$ does not generate elastic stress.
That is, $\bv{\mathcal{E}}_p$ is a stress-free eigen strain \cite{mura_micromechanics_1987}.
Based on the standard theory of linear elasticity, the stress equilibrium equation resulting from the distribution of the plastic strain $\bv{\mathcal{E}}_p$ is formulated as follows:
\begin{align}\label{eq:stressequil2}
    \nabla\cdot \left(\bv{C}_0:\bv{\mathcal{E}}_t\right)
    =\nabla\cdot \left(\bv{C}_0:\bv{\mathcal{E}}_p\right).
\end{align}
This equation represents the linearised stress equilibrium equation, which determines the total Cauchy strain $\bv{\mathcal{E}}_t$ for a specific distribution of plastic strain $\bv{\mathcal{E}}_p$.

Initially, we analyse the stress fields of a screw dislocation.
By substituting the plastic displacement gradient $\bv{\Theta}^S$ of the screw dislocation (\ref{eq:PlasticGradientScrew}) into the definition of plastic strain $\bv{\mathcal{E}}_p$ as outlined in equation (\ref{eq:LinearStrains})$_2$, we obtain the plastic Cauchy strain of the screw dislocation, such that
\begin{align}
  \bv{\mathcal{E}}_p^S = \frac{b}{4\pi}\begin{pmatrix}
    0&0& \dfrac{-y}{x^2+y^2}\\
    0&0& \dfrac{x}{x^2+y^2}\\
    \dfrac{-y}{x^2+y^2}&\dfrac{x}{x^2+y^2}&0
  \end{pmatrix}.
\end{align}
By substituting this plastic strain into the right-hand side of equation (\ref{eq:stressequil2}), we obtain the stress equilibrium equation for screw dislocations.
Notably, the right-hand side of the stress equilibrium equation (\ref{eq:stressequil2}) vanishes identically.
This outcome suggests that the plastic strain $\bv{\mathcal{E}}_p^S$ is in a self-equilibrium state, having no impact on the stress equilibrium equation.
Consequently, the total strain vanishes $\bv{\mathcal{E}}_t^S=\bv{\mathcal{E}}_p^S+\bv{\mathcal{E}}_e^S=\bv{0}$, implying that $\bv{\mathcal{E}}_e^S=-\bv{\mathcal{E}}_p^S$.
However, from the linearised Hooke's law, $\bv{\sigma}^S=\bv{C}_0 : \bv{\mathcal{E}}_e^S=-2\mu \bv{\mathcal{E}}_p^S$.
Consequently, the linearised elastic strain $\bv{\mathcal{E}}_e^S$ and stress field $\bv{\sigma}^S$ of the screw dislocation can be expressed as follows:
\begin{align}
  \bv{\mathcal{E}}_e^S=-\bv{\mathcal{E}}_p^S,
  \quad
  \bv{\sigma}^S={}
  -\frac{\mu b}{2\pi}
  \begin{pmatrix}
    0&0& \dfrac{-y}{x^2+y^2}\\
    0&0& \dfrac{x}{x^2+y^2}\\
    \dfrac{-y}{x^2+y^2}&\dfrac{x}{x^2+y^2}&0
  \end{pmatrix}.
\end{align}
This is the linearised stress field of the screw dislocation obtained from the geometrical theory.
The shear components $(\sigma^S)^{13}=(\sigma^S)^{31}$ and $(\sigma^S)^{23}=(\sigma^S)^{32}$ are non-zero, and are shown in figure \ref{fig.3}(a).
They exhibit singularity along the dislocation line $(x,y)=(0,0)$, diminishing as distance increases.
This outcome aligns quantitatively with the classical Volterra dislocation theory \cite{anderson_theory_2017}.

Similarly, from equation (\ref{eq:PlasticGradientEdge}), the plastic Cauchy strain of the edge dislocation becomes
\begin{align}
  \bv{\mathcal{E}}_p^E = \frac{b}{4\pi}\begin{pmatrix}
    \dfrac{-2y}{x^2+y^2}&\dfrac{x}{x^2+y^2}&0\\
    \dfrac{x}{x^2+y^2}&0&0\\
    0&0&0
  \end{pmatrix}.
\end{align}
By substituting the result into the right-hand side of the stress equilibrium equation (\ref{eq:stressequil2}) and integrating the linear partial differential equation, we obtain the analytical form of the total displacement $\bv{u}$ such that
\begin{align}
  \bv{u} = \frac{b}{2\pi}\begin{pmatrix}
    -\dfrac{x y}{2(1-\nu)(x^2+y^2)},\,
    \dfrac{1}{4(1-\nu)} \bigg\{\dfrac{x^2-y^2}{x^2+y^2}+(1-2\nu) \log(x^2+y^2)\bigg\},\,0
  \end{pmatrix}.
\end{align}
From the total displacement $\bv{u}$ and definition of strains (\ref{eq:LinearStrains}), the elastic Cauchy strain $\bv{\mathcal{E}}_e^E$ and the resulting stress $\bv{\sigma}^E=\bv{C}_0:\bv{\mathcal{E}}_e^E$ can be obtained, such that
\begin{align}
  \bv{\mathcal{E}}_e^E=&{}
  \dfrac{b}{4\pi(1-\nu)}\begin{pmatrix}
    y\dfrac{(3-2\nu) x^2+(1-2\nu) y^2}{(x^2+y^2)^2}&
    \dfrac{x(y^2-x^2)}{(x^2+y^2)^2}&
    0\\
    \dfrac{x (y^2-x^2)}{(x^2+y^2)^2}&
    \dfrac{-(1+2 \nu)x^2 y+(1-2 \nu) y^3}{(x^2+y^2)^2}&
    0\\
    0&0&0
  \end{pmatrix},
  \\
  \bv{\sigma}^E=&{}
    -\frac{\mu b}{2\pi (1-\nu)}\begin{pmatrix}
      -\dfrac{y(3 x^2+y^2)}{(x^2+y^2)^2}&
      \dfrac{x (x^2-y^2)}{(x^2+y^2)^2}&
      0\\
      \dfrac{x (x^2-y^2)}{(x^2+y^2)^2}&
      \dfrac{y (x^2-y^2)}{(x^2+y^2)^2}&
      0\\
      0&0&-\dfrac{2\nu y}{x^2+y^2}
  \end{pmatrix}.
\end{align}
As shown in the above equations, the normal components $(\sigma^E)^{11},\ (\sigma^E)^{22}, (\sigma^E)^{33}$ and the shear component $(\sigma^E)^{12}=(\sigma^E)^{21}$ are non-zero, and plotted in figure \ref{fig.3}(b).
Similar to the stress fields of the screw dislocations, they exhibit singularity along the dislocation line $(x,y)=(0,0)$, diminishing as distance increases.
These observations align perfectly with the strain and stress fields associated with the edge dislocation in the Volterra dislocation theory \cite{anderson_theory_2017}.
Therefore, the elastic stress field of dislocations, derived from principles of differential geometry, closely mirrors the elastic stress predicted by Volterra's theory.
This outcome not only implies that the plastic deformation field of a dislocation shares mathematical similarities with electromagnetism and complex function theory, but also suggests that the analytical methods utilised in these fields can be applied to mechanical analysis of dislocations.

\begin{figure}[!h]
  \centering
  \includegraphics[width=0.95\textwidth]{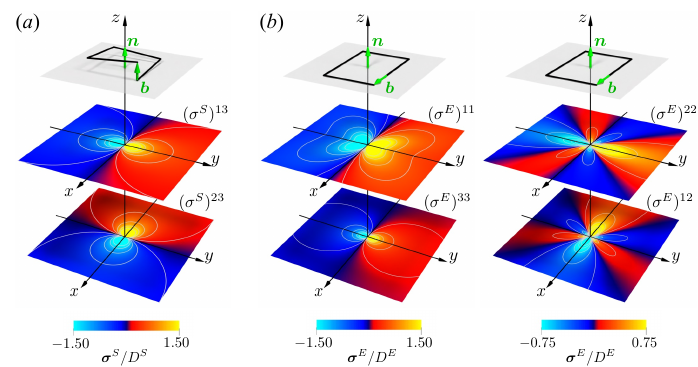}
  \caption{
    Visualisation results of the linear stress fields of screw and edge dislocations.
    (a) shows shear components $(\sigma^S)^{13}$ and $(\sigma^S)^{23}$ of the screw dislocation, while (b) shows normal components $(\sigma^E)^{11}$, $(\sigma^E)^{22}$, $(\sigma^E)^{33}$, and shear component $(\sigma^E)^{12}$ of the edge dislocation, all of which coincident with those of the Volterra dislocation theory.
    The normalisation factor of screw and edge dislocations are $D^S=\mu b/2\pi$ and $D^E=\mu b/2\pi(1-\nu)$, respectively.
  }
  \label{fig.3}
\end{figure}

\subsection{Geometrical origin of elastic deformation}

Finally, we explored the geometrical origin of the elastic deformation formed near the dislocations.
Generally, the connection of a Riemann--Cartan manifold includes torsion and curvature, which are absent in Euclidean geometry.
However, the current configuration $\mathcal{C}$ was achieved by elastically embedding the intermediate configuration $\mathcal{B}$ into the Euclidean space.
The non-Euclidean connection can be viewed as a form of geometrical frustration, hindering the manifold from conforming to Euclidean space.
This suggests that elastic deformation served to alleviate this frustration by eliminating torsion and curvature in the connection.
In our previous study, we demonstrated that the origin of elastic deformation is rooted in geometrical frustration through numerical analysis of the nonlinear stress equilibrium equation \cite{kobayashi_geometrical_2024}.
In this study, we derived analytical expressions for dislocation stress fields using a linear approximation.
While a linear approximation was employed, the foundational role of elastic deformation remained unchanged, and our objective was to validate this assertion within the current framework.

Various geometric quantities can be utilised to characterise the intermediate configuration; however, we focus on the Einstein curvature $\bv{G}$.
The Einstein curvature is defined using the Riemannian metric as follows: $\bv{G} = \bv{R} - \bv{g}_\mathcal{B} R/2$, where $\bv{R} = R^k_{ikj}dx^i \otimes dx^j$ represents the Ricci curvature and $R = R^k_{ikj}g_\mathcal{B}^{ij}$ represents the scalar curvature \cite{nakahara_geometry_2003}.
Both quantities can be derived only using the Riemannian metric $\bv{g}_\mathcal{B}$.
As shown in equation (\ref{eq:IntermediateMetric}), the Riemannian metric for the intermediate configuration can be derived from the plastic deformation gradient $\bv{F}_p$.
Therefore, the Einstein curvature near the dislocation can be obtained from $\bv{F}_p$.
Notably, the Einstein curvature is zero in the Euclidean space, making it a valuable criterion for assessing the degree of geometrical frustration.
This curvature can be expressed by utilising a linear approximation \cite{kroner_allgemeine_1959,teodosiu_elastic_1982,de_wit_view_1981}.
\begin{align}
  \label{eq:einstein curvature}
  \bv{G}=\bv{\eta}+O(b^2),
\end{align}
where $\bv{\eta}$ on the right-hand side denotes the incompatibility tensor in the linear elasticity theory.
This incompatibility tensor can be defined using the dislocation density tensor $\bv{\alpha}$ as $\bv{\eta}=-(\nabla\times \bv{\alpha}+(\nabla\times\bv{\alpha})^T)/2$ \cite{kroner_allgemeine_1959}.
However, the relationship between $\bv{\eta}$ and the elastic strain $\bv{\mathcal{E}}_e$ is well-established \cite{kroner_allgemeine_1959,teodosiu_elastic_1982} and is expressed as
\begin{align}
  \label{eq:strain compatibility equation}
  \nabla\times \bv{\mathcal{E}}_e\times \nabla = -\bv{\eta}.
\end{align}
Based on the results, the linearised elastic strain $\bv{\mathcal{E}}_e$ represents the incompatibility tensor $\bv{\eta}$.
As $\bv{\eta}$ is a linear approximation of the Einstein curvature $\bv{G}$, which is a geometric frustration, the direct origin of the elastic strain $\bv{\mathcal{E}}_e$ can be considered a geometrical frustration itself.

\section{Conclusion}

Conventional theoretical analyses of dislocations have predominantly focused primarily on elastic fields, with minimal consideration given to plastic deformation.
This lack of attention can be attributed to the fact that both plastic and elastic deformations of dislocations are typically confined within Euclidean space, where the mathematical distinction between these deformation modes remains unclear.
In contrast, the present geometric theory extends the kinematics of dislocations to a broader mathematical framework known as the Riemann--Cartan manifold.
This extension allows for the mathematical separation of plastic and elastic deformations, facilitating a more detailed analysis of the characteristics of plastic deformation.
The results obtained in this study are summarised as follows:

\begin{enumerate}
    \item
    The geometrical theory of dislocation extended the kinematics of dislocations to Riemann--Cartan manifolds through the multiplicative decomposition of the total deformation $\bv{F}_t$ into plastic deformation $\bv{F}_p$ and elastic deformation $\bv{F}_e$. The intermediate state, which represents only the plastic deformation, was obtained by solving the Cartan's first structure equation for a specific dislocation density $\bv{\alpha}^i$. In this context, the Helmholtz decomposition enabled the extraction of the plastic displacement gradient $\bv{\Theta}^i$, which is fundamental to plastic deformation fields. Consequently, Cartan's equation and the divergence-free condition for the plastic displacement gradient aligned with Amp\`ere 's and Gauss' laws in electromagnetics.
    \item 
    The mathematical equivalence between plastic deformation and electromagnetics allowed for the application of the analytical solution for the static magnetic field, as described by the well-known Biot--Savart law, to the plastic displacement gradient caused by a dislocation. In other words, the concept of the vector potential in electromagnetism could be directly utilised in the plastic mechanics analysis of dislocations. This equivalence in mathematical structures was independent of the dislocation type, allowing similar analytical solutions to be derived for both screw and edge dislocations. These findings were summarised in Theorems \ref{th:ScrewDislocation} and \ref{th:EdgeDislocation}.
    \item 
    Furthermore, the governing equations for the plastic deformation of dislocations aligned with the Cauchy--Riemann equations in complex function theory. This correspondence suggested that the potential for plastic deformation could be expressed as a complex function, as demonstrated in theorem \ref{thm:PlasticPotential}. When expressed in polar coordinates, the real part of the complex potential resulted in a multivalued function discretised in integer multiples of the Burgers vector. This indicated that dislocations were not merely geometrical defects resulting from changes in the Riemannian metric but possessed distinct differential topological properties.
    \item 
    The mathematical basis for the equivalence between the equations describing the plastic deformation of dislocations --- particularly, Cartan's equation and the divergence-free condition --- and those in electromagnetism and complex function theory, relied on the fact that the plastic displacement gradient possessed a mathematical structure known as an $\mathbb{R}^3$-valued 1-form. Consequently, the plastic deformation of dislocations was distributed in a plane orthogonal to the dislocation line, with the in-plane deformation exhibiting the properties of an angle-preserving conformal map. The consistent orthogonality properties of plastic deformation, regardless of dislocation type, were seen as a reflection of the underlying equivalence in mathematical structure.
    \item 
    Moreover, the plastic deformation of the dislocation obtained through the aforementioned analysis generated an elastic stress field around the dislocation. This mechanical field was established by solving the stress equilibrium equation, with a linear approximation being utilised to analytically derive the stress field. The outcomes closely aligned with those predicted by the Volterra dislocation theory, underscoring the quantitative accuracy of the geometric analysis of dislocations. Furthermore, the direct source of the stress field surrounding the dislocation can be attributed to the geometric frustration characterised by the Einstein curvature.
\end{enumerate}

\section*{Acknowledgment}
This research was supported in part by JST, PRESTO Grant Number JPMJPR1997 Japan, and JSPS KAKENHI Grant Numbers JP18H05481 and JP23K13221.

\bibliographystyle{RS}
% \bibliography{BiotSavartDislocation.bib}

\end{document}